\documentclass[10pt,english]{article} 

\usepackage{amsmath, amsthm, amssymb}
\usepackage{a4, babel, graphicx, listings, fancyhdr, verbatim}
\usepackage{mathrsfs} 
\usepackage{anysize} 
\usepackage{natbib} 
\usepackage[format = hang]{caption}                              

\usepackage{a4}
\usepackage[latin1]{inputenc}
\usepackage{amsmath, amsthm}
\usepackage{graphicx}
\usepackage{amssymb}
\usepackage{verbatim}
\usepackage{listings}
\usepackage{alltt}
\usepackage{babel}
\usepackage{mathrsfs}
\usepackage{ae}
\usepackage{natbib}
\usepackage{booktabs} 

\newtheoremstyle{theorem}
{10pt} 
{10pt} 
{\sl} 
{\parindent} 
{\bf} 
{. } 
{ } 
{} 

\theoremstyle{theorem}

\newtheorem{lemma}{Lemma}

\def\beq{\begin{eqnarray}}
\def\eeq{\end{eqnarray}}

\def\beqn{\begin{eqnarray*}}  
\def\eeqn{\end{eqnarray*}}

\def\E{{\rm E}}
\def\Var{{\rm Var}}

\def\N{{\rm N}}

\def\Pr{{\rm pr}}

\def\quadandquad{\quad \hbox{and} \quad}

\def\hatt{\widehat}
\def\tilda{\widetilde}

\def\sumin{\sum_{i=1}^n}

\def\eps{\varepsilon}
\def\half{\hbox{$1\over2$}}

\def\rootn{\sqrt{n}}
\def\data{{\rm data}}
\def\obs{{\rm obs}}
\def\midd{\,|\,}
\def\tr{{\rm t}}

\def\calU{{\cal U}}

\def\wide{{\rm wide}}
\def\rr{{\rm rr}}
\def\rrr{{\rm rrr}}
\def\conf{{\rm conf}}
\def\risk{{\rm risk}}

\def\fric{{\rm FRIC}}
\def\afric{{\rm AFRIC}}

\def\mse{{\rm mse}}
\def\rmse{{\rm rmse}}
\def\rss{{\rm rss}}

\def\AIC{{\rm AIC}}

\def\Tr{{\rm Tr}}

\def\mse{{\rm mse}}

\def\cc{{\rm cc}}

\numberwithin{equation}{section} 
\numberwithin{figure}{section}
\numberwithin{table}{section}
\usepackage{hyperref}

\title{The Focused Relative Risk Information Criterion \\ 
   for Variable Selection in Linear Regression}

\def\heute{{July 2020}}
\date{\heute}

\begin{document}


\maketitle

\centerline{\large\bf Nils Lid Hjort}

\medskip 
\centerline{\bf Department of Mathematics, University of Oslo}

\begin{abstract}
\noindent
\small{This paper motivates and develops 
a novel and focused approach to variable selection in 
linear regression models. For estimating the 
regression mean $\mu=\E\,(Y\midd x_0)$, for the covariate 
vector of a given individual, there is a list of
competing estimators, say $\hatt\mu_S$ for each 
submodel $S$. Exact expressions are found for 
the relative mean squared error risks, when compared 
to the widest model available, say $\mse_S/\mse_\wide$.
The theory of confidence distributions is used for 
accurate assessments of these relative risks. 
This leads to certain Focused Relative Risk Information Criterion 
scores, and associated FRIC plots and FRIC tables, 
as well as to Confidence plots to exhibit 
the confidence the data give in the submodels. 
The machinery is extended to handle many focus 
parameters at the same time, with appropriate 
averaged FRIC scores. The particular case where 
all available covariate vectors have equal 
importance yields a new overall criterion for 
variable selection, balancing complexity and fit
in a natural fashion. A connection to the Mallows 
criterion is demonstrated, leading also to 
natural modifications of the latter. 
The FRIC and AFRIC strategies are illustrated for real data.  

}
\end{abstract}



\section{Introduction and summary}
\label{section:intro}

Mrs.~Jones is pregnant (again). 
She is white, 40 years old, of average weight 60 kg
before pregnancy, and a smoker. 
What is the expected birthweight of her child-to-come? 

To address this and similar questions, involving comparison 
and ranking of many submodels of a given wide regression
model, we use a dataset on $n=189$ mothers and babies, 
discussed and analysed in \citet[Ch.~2]{ClaeskensHjort08},  
and apply linear regression for the birthweight $y$
in terms of five covariates 
$x_1$ (age), 
$x_2$ (weight in kg before pregnancy),
$x_3$ (indicator for smoker or not),
$x_4$ (ethnicity indicator 1), 
$x_5$ (ethnicity indicator 2),  
and where `white' corresponds to these two being zero,
i.e.~not belonging to the two other ethnic groups in question.  
The full linear regression model has 
\beq
\label{eq:widemodeljones}
y_i=\beta_0+\beta_1x_{i,1}+\cdots+\beta_5x_{i,5}+\eps_i 
   \quad {\rm for\ }i=1,\ldots,n, 
\eeq 
with the $\eps_i$ assumed i.i.d.~from a normal $\N(0,\sigma^2)$. 
The task is to estimate $\mu=\E\,(Y\midd x_0)$, 
with $x_0=(x_{0,1},\ldots,x_{0,5})$ the covariate vector
for Mrs.~Jones. For each of the $2^5=32$ submodels, corresponding 
to taking covariates in and out of the above
regression equation, there is a point estimate,
say $\hatt\mu_S$, with $S$ a subset of $\{1,\ldots,5\}$,
and these are plotted on the vertical axis of
the {\it FRIC plot} of Figure \ref{figure:figure11}, 
along with submodel-specific 80\% confidence intervals, 
to indicate their relative precision. 
The crucial new and extra aspect of the plot are 
the {\it FRIC scores}, plotted on the horizontal axis.
These {\it Focused Relative Risk Information Criterion}
scores are accurately constructed estimates 
of the relative risk, the ratio of mean squared errors,
relative to the wide model, i.e.
\beq
\label{eq:hereisrr} 
\rr_S
={\mse_S\over \mse_\wide}
={\E\,(\hatt\mu_S-\mu)^2\over \E\,(\hatt\mu_\wide-\mu)^2}
\quad {\rm for\ }S{\rm\ subset\ of\ }\{1,\ldots,5\}. 
\eeq 
Thus models with FRIC scores below 1 are judged 
to be better than the full wide model 
of (\ref{eq:widemodeljones}), for the given 
purpose of estimating the mean well for the 
specified covariate vector; FRIC values higher than 1 indicate
that the submodel does a worse job than the wide model itself. 
Only the wide model based estimator $\hatt\mu_\wide$ 
is guaranteed to have zero bias, so the statistical game 
is to use the data to hunt for submodels leading 
to lower variances and biases not far from zero.  

\begin{figure}[h]
\centering
\includegraphics[scale=0.55]{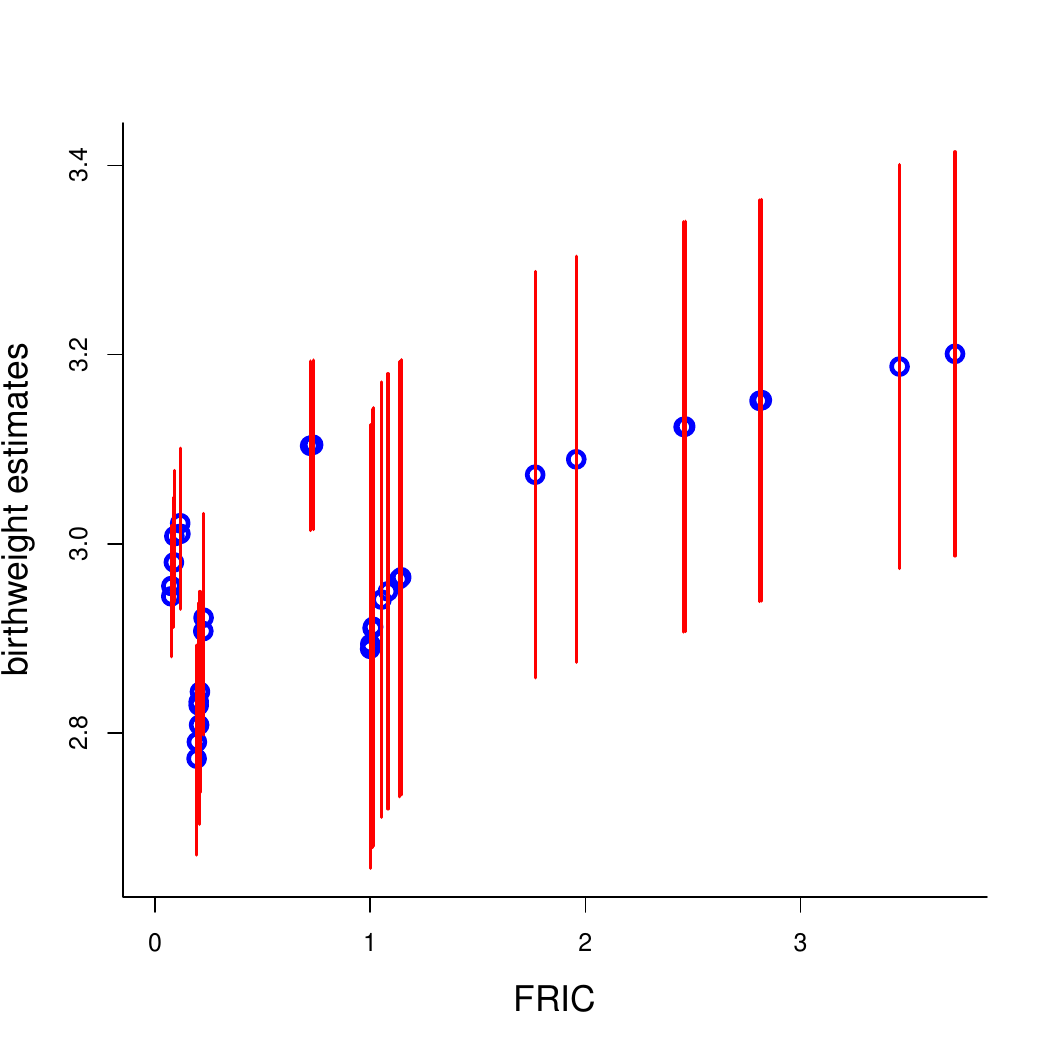}
\caption{FRIC plot for the $2^5=32$ models for estimating
the birthweight of the child-to-come, for Mrs.~Jones
(white, age 40, 60 kg, smoker). The FRIC scores 
are estimates of the relative risks $\rr_S=\mse_S/\mse_\wide$
of (\ref{eq:hereisrr}); the blue circles are the associated
point estimates; and the vertical lines are submodel-based 
80\% confidence intervals. Here 16 submodels have
FRIC scores smaller than 1 and are judged to be better
than the wide model. 
See Table \ref{table:table11} for identification of 
the best models and their FRIC scores and estimates.} 
\label{figure:figure11}
\end{figure}

\subsection{FRIC plots and tables for relative risks}

Inside the natural framework of linear regression models,
basic formulae for the required quantities are worked 
out in Section \ref{section:basics}. The operating assumption
is that the wide model, with all covariates on board, 
is in force. The denominator of (\ref{eq:hereisrr}) 
is standard, whereas more care is needed to find 
a fruitful formula for the numerator, since submodels 
will carry biases. In Section \ref{section:rrwithCD} 
we construct several natural estimators for the relative risk 
quantities $\rr_S$, and such are indeed used to produce 
Figure \ref{figure:figure11}. 
Importantly, as part of this development we construct 
informative and exact {\it confidence distributions} 
for the relative risks. These are data driven 
cumulative distribution functions $C_S(\rr_S,\data)$ 
with the property 
\beq
\label{eq:hereisCD}
\Pr_{\beta,\sigma}\{\rr_S\colon C_S(\rr_S,\data)\le\alpha\}=\alpha
   \quad {\rm for\ all\ }\alpha\in(0,1). 
\eeq 
In (\ref{eq:hereisCD}) the `data' are random, 
with distribution governed by (\ref{eq:widemodeljones}), 
and the identity, being valid for all $(\beta,\sigma)$, 
secures that accurate confidence intervals
at all levels can be read off for the relative risks 
$\mse_S/\mse_\wide$. As we show in Section \ref{section:rrwithCD},
these confidence distributions also start out with certain
pointmasses at observable minimum positions, say $\rr_{S,\min}$. 
Natural confidence intervals for the relative risks $\rr_S$
therefore do {\it not} take the usual form of 
$\hatt\rr_S\pm z_\alpha$, say, an estimate along with a 
plus-minus estimated error. Instead, confidence intervals
induced by the confidence distributions (\ref{eq:hereisCD})
are often asymmetric, and sometimes start out at the 
indicated minimum possible value; 
see indeed Figure \ref{figure:figure15}.

\begin{figure}[h]
\centering
\includegraphics[scale=0.55]{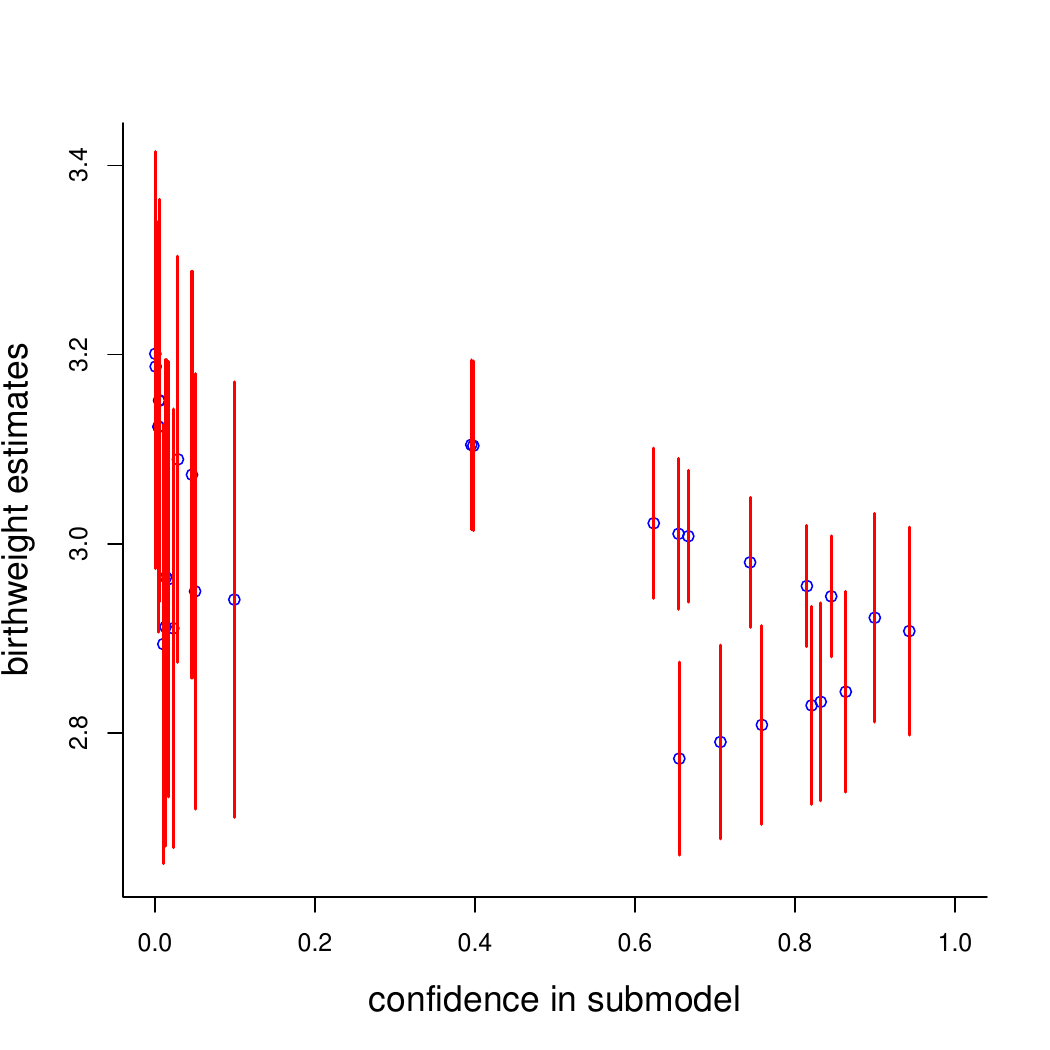}
\caption{Confidence FRIC plots for the $2^5=32$ models for estimating
the birthweight of the child-to-come, for Mrs.~Jones
(white, age 40, 60 kg, smoker). The $\conf(S)$ values
are also p-values for testing $\mse_S\le\mse_\wide$,
so submodels to the left in the figure are found
not useful for the task of estimating the focus parameter,
the mean $\mu=\E\,(Y\midd x_0)$. Submodels to the far right
are those in which we place trust in their ability to 
do better than the wide model. 
See Table \ref{table:table11} for identification of 
the best models and their $\conf(S)$ scores and estimates.} 
\label{figure:figure13}
\end{figure}

Two tools flowing from these insights are as follows. 
First, one may use the {\it median confidence estimators} 
$\hatt\rr_S^{0.50}=C_S^{-1}(0.50,\data)$ as FRIC scores 
when constructing the FRIC plots. This bypasses certain 
difficulties otherwise encountered, related to 
a decision of whether one needs to truncate nonnegative
estimates of squared biases to zero or not; 
see the discussion of \citet{CunenHjort20a}. 
Second, the epistemic confidence 
in a submodel's ability to work better than the wide model 
can be read off separately, i.e. 
\beq
\label{eq:conf}
\conf(S)
=C_S(1,\data)
=\Pr^*\{\mse_S/\mse_\wide\le 1\}
   \quad {\rm for\ }S\in\{1,\ldots,5\}. 
\eeq 
Here $\Pr^*$ denotes the post-data epistemic probability
placed on the event $\mse_S\le\mse_\wide$; 
see further discussion in Section \ref{section:rrwithCD}.
Good models, again for the focused purpose of 
estimating a given mean well, are those for which 
these $\conf(S)$ scores are high; 
conversely, candidate models where that score 
is low can be seen as not doing well for the purpose. 
Indeed, $\conf(S)$ is the p-value for
testing the hypothesis that $\mse_S\le\mse_\wide$. 
Figure \ref{figure:figure13} displays such a 
Confidence plot, with these confidence scores 
for the case of Mrs.~Jones. 
 

\begin{table}[h]
\small
\begin{center}
\begin{tabular}{c | c | c |  c |  c     c  c | c | c | c }
   &FRIC table:   &        &     &     & & Conf table: \\
   &in-or-out     & $\hatt\mu_S$ & FRIC & conf & & in-our-out & $\hatt\mu_S$ & FRIC & conf \\
\midrule 
 1 &   0 0 0 0 0 &  2.945 & 0.076 &    0.845 & & 0 0 1 1 1 &  2.908 &  0.224 & 0.943 \\
 2 &   0 1 0 0 0 &  2.956 & 0.076 &    0.815 & & 0 1 1 1 1 &  2.922 &  0.226 & 0.900 \\
 3 &   0 0 0 1 0 &  2.981 & 0.088 &    0.744 & & 0 1 1 1 0 &  2.844 &  0.209 & 0.863 \\
 4 &   0 1 0 1 0 &  3.008 & 0.090 &    0.667 & & 0 0 0 0 0 &  2.945 &  0.076 & 0.845 \\
 5 &   0 0 0 0 1 &  3.022 & 0.117 &    0.623 & & 0 1 1 0 1 &  2.833 &  0.203 & 0.832 \\
 6 &   0 1 0 0 1 &  3.011 & 0.118 &    0.654 & & 0 0 1 0 1 &  2.829 &  0.203 & 0.821 \\
 7 &   0 0 1 0 0 &  2.773 & 0.193 &    0.655 & & 0 1 0 0 0 &  2.956 &  0.076 & 0.815 \\
 8 &   0 1 1 0 0 &  2.791 & 0.195 &    0.707 & & 0 0 1 1 0 &  2.809 &  0.205 & 0.759 \\
 9 &   0 0 1 0 1 &  2.829 & 0.203 &    0.821 & & 0 0 0 1 0 &  2.981 &  0.088 & 0.744 \\
10 &   0 1 1 0 1 &  2.833 & 0.203 &    0.832 & & 0 1 1 0 0 &  2.791 &  0.195 & 0.707 \\
17 &   1 1 1 1 1 &  2.889 & 1.000 &     --   & & --        &  --    &  -- &  --   \\
\end{tabular}
\end{center}
\caption{$\fric_S$ (left part) and $\conf(S)$ (right part) 
tables for Mrs.~Jones, where the task is precise estimation 
of $\mu=\E\,(Y\midd x_0)$ for her given covariate vector $x_0$. 
The left part shows the ten best submodels 
of the $2^5=32$ candidate models, sorted by FRIC scores,
see Figure \ref{figure:figure11};  
the right part the ten best submodels, sorted
by the $\conf(S)$ scores, see Figure \ref{figure:figure13}. 
Models with low FRIC scores tend to have high $\conf(S)$ scores,
and vice versa. The in-our-out columns indicate
submodels by presence-absence of covariates.} 
\label{table:table11}
\end{table}

The FRIC plot and Confidence plot are already informative,
displaying the many submodel estimates along with
accurate information about estimated relative risks
and about degree of confidence. One might proceed
by taking an average of the the say ten best 
point estimates, or with more elaboration form
model average estimators of the form 
$\hatt\mu^*=\sum_S v(S)\hatt\mu_S$, with weights
$v(S)$ normalised to sum to one, 
and chosen higher for the best models, according to 
either the FRIC or the $\conf(S)$ scores. 

Further pertinent statistical information may be
gleaned from inspection of which models do well,
however, and also from examination of those which
do badly. Along with FRIC plots and Confidence plots, 
therefore, our R programme produces tables 
of the form shown here in Table \ref{table:table11}. 
Interestingly, for the covariate information $x_0$
for Mrs.~Jones, age doesn't matter, as the 
covariate $x_1$ does not enter any of the ten best
candidate models. 

Naturally, an essential point for the FRIC and Confidence plots 
and tables is that they are {\it focused}; they 
work for one given covariate vector $x_0$ at the 
time. They deliver personalised model ranking lists
and top estimates for the focused purpose. 
If Mrs.~Jones were a non-smoker, we put $x_3$ 
equal to 0 instead of 1, run our FRIC programmes
again, and find new plots and tables, 
new top models, and of course new estimates of 
the pertinent $\mu=\E\,(Y\midd x_0)$ parameter. 
For the contrafactual non-smoking Mrs.~Jones 
one then learns that age $x_1$ matters, after all,
and that the smoker indicator $x_3$ enters several
more of the top ranked models -- also, the average
estimate over the ten top models becomes 
$\hatt\mu=3.197$ kg, rather than the 2.915 kg 
she can expect as a smoker 
(of the same age, weight, and ethnicity). 

These general points are broadly valid for the 
growing list of FIC methods in the literature, where 
the Focused Information Criteria, introduced 
and developed in \citet{ClaeskensHjort03, ClaeskensHjort08}, 
have such focused aims, of finding the best models
for given focus parameters. These FIC type strategies 
have later been generalised and extended in several
directions, and for new classes of models; 
see \citet*{ClaeskensCunenHjort19} for a review.
The present paper is 
different in that the FRIC methods developed here are 
{\it accurate and exact}, with no large-sample 
approximations involved; in a sense finite-sample
corrections are not required as such are 
already built into the paper's exact formulae. 
This is partly due to the relative simplicity of 
the linear normal regression model and its submodels.
The exact finite-sample formulae worked out in 
sections below continue to be valid, as good 
approximations, in the case of error distributions
not being exactly normal; the more vital assumptions
are those of the linear mean, the constant variance,
and independence. 

Similar FRIC methods might be set up for 
addressing, approximating, estimating, and ranking 
relative risks, say $\mse_S/\mse_\wide$ for 
general focus parameters in generalised linear models.
Such constructions would need some elements of 
large-sample approximations, however, for 
variances, biases, and then the assessment
of estimates for these again; 
see indeed \citet{CunenHjort20a} for such developments. 

\subsection{The present paper}

As explained and exemplified with the case of Mrs.~Jones
above, the development in Sections \ref{section:basics}
and \ref{section:rrwithCD} concerns {\it given focus parameters},
giving in particular FRIC formulae for estimating 
relative risks and also focused confidence distributions,
for each potential $\mu(x_0)=\E\,(Y\midd x_0)$.
The apparatus can also be used for assessing and 
ranking all candidate models for their ability 
to estimate a particular regression parameter well. 
Importantly, the machinery can be extended to 
handling an ensemble of focus parameters jointly, 
as when one needs submodels working well for 
certain regions of covariates, like the stratum of 
all pregnant women of age below twenty in the 
birthweight study. Such extensions, leading to 
certain AFRIC scores, averaging over focus parameters, 
are developed in Section \ref{section:afric}. 
A notable special case is when all available 
covariate vectors are considered equally important, 
leading to particularly illuminating AFRIC formulae,
balancing overall fit with model complexity. 
This `unfocused AFRIC' is demonstrated to have
a connection to the Mallows Cp criterion 
in Section \ref{section:mallows}, which also
leads to both a modification of and additional
insights into the Mallows statistic. 
The AFRIC is illustrated using the birthweight dataset. 

Our paper is then rounded off with a list of concluding remarks
in Section \ref{section:concluding}, some pointing
to further research issues. These touch on model averaging, 
on generalisations to more complex regression models,
and on certain simplifications valid under special circumstances. 

Whereas the present paper presents novel variable 
selection methods, via focused assessment of 
relative risks, it is clear that other FIC based 
model selection methods also can work well; see the 
references pointed to above, and the review by 
\citet*{ClaeskensCunenHjort19}. There is an 
enormous literature on model selection in general 
and variable and subset selection methods for regressions
in particular, see \citet{ClaeskensHjort08} 
and e.g.~\citet{Miller02}. For the literature 
on confidence distributions and their many related
themes, see \citet{XieSingh13}, \citet{SchwederHjort16}, \citet{HjortSchweder18}. 

\section{Submodels and relative risks} 
\label{section:basics}

Consider the classic linear regression model 
\beq
\label{eq:widemodel}
y_i=x_i^\tr\beta+\eps_i
   =\beta_1x_{i,1}+\cdots+\beta_px_{i,p}+\eps_i 
   \quad {\rm for\ }i=1,\ldots,n, 
\eeq 
with the error terms $\eps_i$ being i.i.d.~from $\N(0,\sigma^2)$, 
and with covariate vectors $x_i$ of dimension $p$. 
We call this the wide model, as we shall be
considering the submodels associated with selecting 
covariates for inclusion and exclusion. In the present
and the following section we develop the necessary 
formalism for handling all submodels, find 
explicit formulae for the mean squared errors 
$\mse_S$ and $\mse_\wide$ of (\ref{eq:hereisrr}), 
construct several natural estimators for 
the relative risks $\mse_S/\mse_\wide$, and also 
develop the confidence distributions tools 
pointed to in (\ref{eq:hereisCD}).  
This is the basis for the 
FRIC and Confidence plots of the type 
shown in Figures \ref{figure:figure11} and \ref{figure:figure13}, 
and also of FRIC and Confidence tables 
as in Table \ref{table:table11}. 

The least squares estimator for $\beta$ in the
wide model is 
\beqn
\hatt\beta_\wide=\Sigma_n^{-1}n^{-1}\sumin x_iy_i,
\quad {\rm with} \quad 
\Sigma_n=n^{-1}\sumin x_ix_i^\tr,  
\eeqn 
and with this $p\times p$ matrix assumed to have
full rank. The traditional and unbiased estimator of the
residual variance is 
$\hatt\sigma^2=(n-p)^{-1}\sumin (y_i-x_i^\tr\hatt\beta)^2$,
and under normality we all know that 
\beqn
\hatt\beta_\wide\sim\N_p(\beta,(\sigma^2/n)\Sigma_n^{-1}),
\quad {\rm independent\ of} \quad 
\hatt\sigma^2/\sigma^2\sim\chi^2_m/m, 
\eeqn 
with the chi-squared with degrees of freedom $m=n-p$. 

Now consider a submodel, employing only 
$x_{i,j}$ for $j\in S$, a subset of $\{1,\ldots,p\}$. 
We need some notation for handling these submodels
and their ensuing estimators. Let $\pi_S$ be the 
projection function associated with $S$,
so that $\pi_S u=u_S$ picks the elements 
$u_j$ of $u=(u_1,\ldots,u_p)$ for which $j\in S$; 
$\pi_S$ can then be written as a matrix of size 
$|S|\times p$ with 0s and 1s, with $|S|$ the
number of elements in $S$.  The $\pi_S$ matrix
can also be seen as containing the rows $j$ 
of the $p\times p$ identity matrix corresponding 
to $j\in S$. 
The submodel in question, indexed by $S$, 
has mean function $x_{i,S}^\tr\beta_S$, 
and least squares estimator 
\beq
\label{eq:betahatS}
\hatt\beta_S=\Sigma_{n,S}^{-1}n^{-1}\sumin x_{i,S}y_i,
\quad {\rm with} \quad 
\Sigma_{n,S}=n^{-1}\sumin x_{i,S}x_{i,S}^\tr
   =\pi_S\Sigma_n\pi_S^\tr.   
\eeq

Variable selection may serve several purposes. Here we 
shall focus attention on the task of estimating 
{\it a given focus parameter}, namely 
$\mu=\E\,(Y\midd x_0)=x_0^\tr\beta$, 
for a given covariate vector $x_0$,
perhaps associated with some given individual 
or object. This is exemplified by the pregnant
Mrs.~Jones featured in our introduction section. 
Thus there are candidate estimators 
$\hatt\mu_S=x_{0,S}^\tr\hatt\beta_S$, one for each $S$. 
The mean squared error (mse) of the wide model 
estimator is easily written down, since it is 
unbiased, and one finds 
\beq
\label{eq:msewide}
\mse_\wide=\Var\,x_0^\tr\hatt\beta
   =(\sigma^2/n) x_0^\tr\Sigma_n^{-1}x_0. 
\eeq 
Setting up a clear formula for the mse of the $S$ based 
estimator is somewhat more tricky. 
Its variance is 
$\Var\,\hatt\mu_S=(\sigma^2/n)x_{0,S}^\tr\Sigma_{n,S}^{-1}x_{0,S}$, 
but some algebraic efforts are needed 
to sort out its bias in a fruitful fashion. Write 
\beq
\label{eq:SigmanSblocked}
\Sigma_n=\begin{pmatrix}
\Sigma_{00,S} &\Sigma_{01,S} \\
\Sigma_{10,S} &\Sigma_{11,S} \end{pmatrix}, 
\quad {\rm with\ inverse} \quad 
\Sigma_n^{-1}=\begin{pmatrix}
\Sigma^{00,S} &\Sigma^{01,S} \\
\Sigma^{10,S} &\Sigma^{11,S} \end{pmatrix}, 
\eeq 
with the relevant blocking into blocks inside and outside $S$.
Thus $\Sigma_{00,S}$ is the same as $\Sigma_{n,S}$ 
of (\ref{eq:betahatS}), etc.; also, the matrix 
\beq
\label{eq:hereisQS}
Q_S=\Sigma^{11,S}=(\Sigma_{11,S}-\Sigma_{10,S}\Sigma_{00,S}^{-1}\Sigma_{01,S})^{-1}
\eeq 
will be needed below, of size $|S^c|\times|S^c|$, 
with $S^c$ the complement set of $S$. 

\begin{lemma}{{\rm 
For estimating $\mu=x_0^\tr\beta$, the submodel based 
estimator $\hatt\mu_S=x_{0,S}^\tr\hatt\beta_S$ has bias 
$\omega_S^\tr\beta_{S^c}$, with 
$\omega_S=\Sigma_{10,S}\Sigma_{S,00}^{-1}x_{0,S}-x_{0,S^c}$, 
of length $|S^c|=p-|S|$. The identity 
$x_0^\tr\Sigma_n^{-1}x_0=x_{0,S}^\tr\Sigma_{n,S}^{-1}x_{0,S}
   +\omega_S^\tr Q_S\omega_S$ holds.  
The mean squared error of $\hatt\mu_S$ can be expressed as }}
\beqn
\mse_S
&=&(\sigma^2/n)x_{0,S}^\tr\Sigma_{n,S}^{-1}x_{0,S}
   +(\omega_S^\tr\beta_{S^c})^2 \\
&= &(\sigma^2/n)
   \{x_{0,S}^\tr\Sigma_{n,S}^{-1}x_{0,S}
   +n(\omega_S^\tr\beta_{S^c})^2/\sigma^2\} \\
&=&(\sigma^2/n)
   \{x_0^\tr\Sigma_n^{-1}x_0 
   +n(\omega_S^\tr\beta_{S^c})^2/\sigma^2 - \omega_S^\tr Q_S\omega_S\}.
\eeqn 
\label{lemma:lemma11}
\end{lemma} 

\vspace{-1.0cm}

\begin{proof}
We start from 
\beqn
\E\,\hatt\beta_S=\Sigma_{n,S}^{-1}n^{-1}\sumin x_{i,S} 
   (x_{i,S}\beta_S+x_{i,S^c}\beta_{S^c})
   =\Sigma_{00,S}^{-1}(\Sigma_{00,S}\beta_S+\Sigma_{01,S}\beta_{S^c}), 
\eeqn 
which leads to the bias expression 
\beqn
\E\,\hatt\mu_S-\mu
&=&x_{0,S}^\tr(\beta_S+\Sigma_{00,S}^{-1}\Sigma_{01,S}\beta_{S^c})
    -x_{0,S}^\tr\beta_S-x_{0,S^c}^\tr\beta_{S^c} \\ 
&=& (x_{0,S}^\tr\Sigma_{00,S}^{-1}\Sigma_{01,S}-x_{0,S^c}) \beta_{S^c}, 
\eeqn 
proving the first assertion. The identity decomposing 
$x_0^\tr\Sigma_n^{-1}x_0$ into two parts 
is proven via the required algebraic manipulations
and patience; also, the expressions for $\mse_S$ follow readily.
\end{proof} 

The lemma, combined with (\ref{eq:msewide}), 
leads to the relative risk expression 
\beq
\label{eq:rrfirst}
\rr_S={\mse_S\over \mse_\wide}
   ={x_{0,S}^\tr\Sigma_{n,S}^{-1}x_{0,S}+n\lambda_S^2
   \over x_0^\tr\Sigma_n^{-1}x_0}, 
\quad {\rm where} \quad
\lambda_S=\omega_S^\tr\beta_{S^c}/\sigma, 
\eeq 
the point also being that the $\sigma^2/n$ term cancels out.
In particular, submodel $S$ is better than the wide model
provided 
\beq
\label{eq:betterornot} 
|\lambda_S|=|\omega_S^\tr\beta_{S^c}|/\sigma
<(x_0^\tr\Sigma_n^{-1}x_0-x_{0,S}^\tr\Sigma_{n,S}^{-1}x_{0,S})^{1/2} / \rootn
= (\omega_S^\tr Q_S\omega_S)^{1/2}/\rootn. 
\eeq 
Note that with more data it typically becomes 
increasingly harder for a simple model to beat the wide model,
but for moderate datasets it is perfectly 
possible for (\ref{eq:betterornot}) to hold, 
depending on the size of the implied bias $\omega_S^\tr\beta_{S^c}$. 
It is noteworthy that if $\beta_{S^c}$ is close to 
being orthogonal to $\omega_S$, then $|\lambda_S|$
is small and the submodel in question may do 
very well, even if $\beta_{S^c}$ is far from zero, 
i.e.~even if the submodel is far from being correct as such. 
This is a version of the classical 
bias-versus-variance balancing game,
which now actively depends on the given $x_0$ 
under focus.   

Of course the statistician cannot directly know 
whether (\ref{eq:betterornot}) holds or not, 
since $\beta_{S^c}/\sigma$ is unknown,  
but clear estimators may be constructed for 
the risk ratios (\ref{eq:rrfirst}), of the form 
\beq
\label{eq:genericfric}
\fric_S
=\hatt\rr_S
={x_{0,S}^\tr\Sigma_{n,S}^{-1}x_{0,S}+\hatt\Lambda_S
   \over x_0^\tr\Sigma_n^{-1}x_0}, 
\quad {\rm with\ }\hatt\Lambda_S{\rm\ estimating\ }
  \Lambda_S=n\lambda_S^2.
\eeq 
There are several natural such, as we come back to below. 
For each of these constructions, we have a FRIC plot 
and a FRIC table, as explained and exemplified 
in the introduction section.  

It is fruitful to factor in the particular variance 
\beqn
\tau_S^2=\Var\,(\omega_S^\tr\rootn\hatt\beta_{\wide,S})
   =\sigma^2\omega_S^\tr Q_S\omega_S,    
\eeqn 
with $Q_S=\Sigma^{11,S}$ the matrix of (\ref{eq:hereisQS}). 
Now consider the parameter 
\beq
\label{eq:hereiskappaS}
\kappa_S={\omega_S^\tr\rootn\beta_{S^c}\over \tau_S}
   ={\rootn\omega_S^\tr\beta_{S^c}\over (\omega_S^\tr Q_S\omega_S)^{1/2}\sigma}, 
\eeq 
with estimator 
\beq
\label{eq:kappakappa}
\hatt\kappa_S={\omega_S^\tr\rootn\hatt\beta_{\wide,S^c}\over \hatt\tau_S}
   ={\omega_S^\tr\rootn\hatt\beta_{\wide,S^c}\over 
   (\omega_S^\tr Q_S\omega_S)^{1/2} \hatt\sigma},
\quad {\rm with\ distribution} \quad 
{\N(\kappa_S,1)\over \hatt\sigma/\sigma}
   \sim t_m(\kappa_S),
\eeq
the noncentral $t_m$ distribution with excentre parameter $\kappa_S$. 
Also, $\hatt\kappa_S^2\sim F_{1,m}(\kappa_S^2)$, 
the noncentral $F$ with degrees of freedom $(1,m)$
and excentre parameter $\kappa_S^2$. 
With this $\kappa_S$ we have the exact relative risk expression 
\beq
\label{eq:rrsecond}
\rr_S
={\mse_S\over \mse_\wide}
={x_{0,S}^\tr\Sigma_{n,S}^{-1}x_{0,S} + \omega_S^\tr Q_S\omega_S\,\kappa_S^2
   \over x_0^\tr\Sigma_n^{-1}x_0}
=1-{\omega_S^\tr Q_S\omega_S\over x_0^\tr\Sigma_n^{-1}x_0}(1-\kappa_S^2), 
\eeq 
cf.~(\ref{eq:rrfirst}). 
In this risk ratio expression all quantities are known,
from the $n\times p$ covariate matrix of the $x_i$, apart from $\kappa_S$. 

For estimating the relative risk ratios (\ref{eq:rrsecond}), 
for comparing and ranking the different submodels, 
and in particular to decide on which submodel
is the best for the given purpose of estimating $x_0^\tr\beta$,
there are a few related options. Using $\hatt\kappa_S^2$
directly leads to overshooting, as 
\beqn
\E\,\hatt\kappa_S^2=\E\,{\kappa_S^2+1\over \hatt\sigma^2/\sigma^2}
   ={m\over m-2}(\kappa_S^2+1). 
\eeqn 
Two natural estimators of the $\kappa_S^2$ term are hence 
\beqn
{m-2\over m}(\hatt\kappa_S^2-1) 
\quadandquad 
\max\Bigl\{{m-2\over m}(\hatt\kappa_S^2-1),0\Bigr\}. 
\eeqn 
The first is unbiased for $\kappa_S^2$, the second 
is the truncated version where negative estimates
of the nonnegative quantity is set to zero. This leads 
to the unbiased FRIC scores, say 
\beq
\label{eq:fricu}
\begin{array}{rcl}
\fric^u_S
&=&\displaystyle
   \Bigl\{ x_{0,S}^\tr\Sigma_{n,S}^{-1}x_{0,S} 
   + {m-2\over m}\omega_S^\tr Q_S\omega_S\,(\hatt\kappa_S^2-1) \Bigr\}
\Big/x_0^\tr\Sigma_n^{-1}x_0 \\ 
&=&\displaystyle
\Bigl[ x_{0,S}^\tr\Sigma_{n,S}^{-1}x_{0,S} 
   + {m-2\over m}
\Bigl\{ {n(\omega_S^\tr\hatt\beta_{\wide,S^c})^2\over \hatt\sigma^2} 
   - \omega_S^\tr Q_S\omega_S\Bigr\} \Bigr]
\Big/x_0^\tr\Sigma_n^{-1}x_0,  
\end{array}
\eeq 
and its truncated cousin, say $\fric_S^t$, 
where the squared bias estimator component,
i.e.~the second term of the numerator here, 
is set to zero in cases where $\hatt\kappa_S^2<1$. 
The probability of this event taking place is 
$F_{1,m}(1,\kappa_S^2)$, which can be as high as 
$\Pr\{\chi^2_1<1\}=0.683$ in the case of zero 
bias and $m=n-p$ high. 
The two related sets of scores $\fric_S^u$ and $\fric_S^t$
otherwise tend to be highly correlated and hence to
produce highly related model rankings. 
For Figures \ref{figure:figure11}--\ref{figure:figure13}
and Table \ref{table:table11}, we have used 
the truncated version of (\ref{eq:fricu}). 

\section{Confidence distributions for the relative risks} 
\label{section:rrwithCD}

Above we were able to derive clear expressions 
for the relative risks $\rr_S=\mse_S/\mse_\wide$, 
as with (\ref{eq:rrsecond}), and then constructed
natural FRIC scores for estimating these, essentially via 
understanding the extent to which $\hatt\kappa_S^2$
of (\ref{eq:kappakappa}) overshoots the $\kappa_S^2$ 
parameter of (\ref{eq:hereiskappaS}). 
These scores have indeed been used for the 
FRIC Figure \ref{figure:figure11} and FRIC Table \ref{table:table11}.

\begin{figure}[h]
\centering
\includegraphics[scale=0.55]{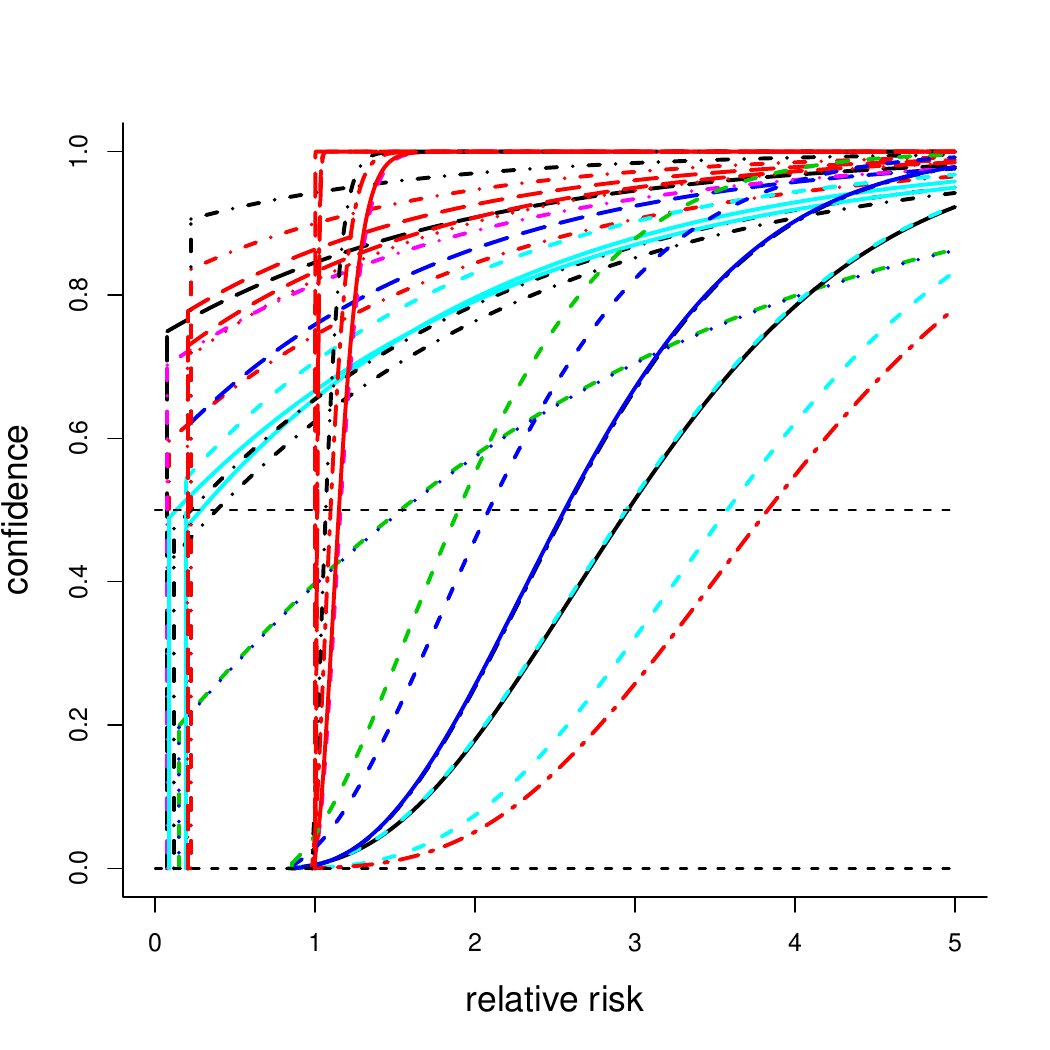}
\caption{Confidence cumulative distribution functions
$C_S(\rr_S)$ for the relative risks $\risk_S/\risk_\wide$, 
as per (\ref{eq:cdforrr}). Those with high confidence
in values below 1 are better for the focused prediction
job for Mrs.~Jones's baby than the wide model. 
The median confidence estimates $\fric_S^{0.50}$ are also
read off from the plot.}
\label{figure:figure15}
\end{figure}

Importantly, we may also set up a clear and exact 
confidence distribution for the $\kappa_S^2$ parameter, 
\beq
\label{eq:cdforkappasq}
C_S^*(\kappa_S^2,\data)
=\Pr_{\kappa_S}\{\hatt\kappa_S^2\ge\hatt\kappa_{S,\obs}^2\}
=1-F_{1,m}(\hatt\kappa_{S,\obs}^2,\kappa_S^2), 
\eeq 
where again $m=n-p$. 
It may be used for more nuanced inference, complete
with accurate confidence intervals at all levels,
which we also exploit below on the more 
canonical $\rr_S$ scale. Note that the confidence 
distribution here has a pointmass at zero, 
of size $C_S^*(0,\data)=1-F_{1,m}(\hatt\kappa_S^2)$, with 
$F_{1,m}(u)$ the cumulative distribution function for 
the ordinary $F$ distribution with degrees of freedom $(1,m)$. 
Furthermore, since 
\beqn
\kappa_S^2
   ={x_0^\tr\Sigma_n^{-1}x_0\,\rr_S-x_{0,S}^\tr\Sigma_{n,S}^{-1}x_{0,S}
   \over \omega_S^\tr Q_S\omega_S}, 
\eeqn 
there is a clear, full, and exact confidence distribution 
for the relative risk ratio of $\rr_S=\mse_S/\mse_\wide$, 
namely
\beq
\label{eq:cdforrr}
C_S(\rr_S,\data)=1-F_{1,m}\Bigl(\hatt\kappa_S^2,
   {x_0^\tr\Sigma_n^{-1}x_0\,\rr_S-x_{0,S}^\tr\Sigma_{n,S}^{-1}x_{0,S}
   \over \omega_S^\tr Q_S\omega_S}\Bigr),   
\eeq 
valid above a well-defined minimum point, 
\beqn
\rr_S\ge \rr_{S,\min}
   =x_{0,S}^\tr\Sigma_{n,S}^{-1}x_{0,S} / x_0^\tr\Sigma_n^{-1} x_0.
\eeqn 
As in \citet{CunenHjort20a} it is useful to 
display all these confidence distributions 
in a diagram, with their different starting points. 
It is also useful to read off all 
\beqn
C_S(1,\data)
=1-F_{1,m}\Bigl(\hatt\kappa_S^2,
   {x_0^\tr\Sigma_n^{-1}x_0 - x_{0,S}^\tr\Sigma_{n,S}^{-1}x_{0,S}
   \over \omega_S^\tr Q_S\omega_S}\Bigr)
=1-F_{1,m}(\hatt\kappa_S^2,1), 
\eeqn 
the epistemic confidence probability that 
$\rmse_S\le\rmse_\wide$. Good models are those where
these probabilities of being better than the wide model
are high, which means $|\hatt\kappa_S|$ ratios sufficiently small. 
Figure \ref{figure:figure13} illustrates this,
where the statistician learns which submodels can 
be expected to be particularly well-working when it
comes to predicting the birthweight of 
Mrs.~Jones's child-to-come. 

The confidence distributions also invite the 
{\it median confidence estimators} 
for the relative risks, yielding one more natural FRIC score,
\beqn
\fric_S^{0.50}
=\hatt\rr_S^{0.50}
=C_S^{-1}(0.50,\data)
=\min\{\rr_S\ge\rr_{S,\min}\colon
   C_S(\rr_S,\data)\ge0.50\}. 
\eeqn 
In cases where the confidence pointmass at the minimum value, 
i.e.~$C_S(\rr_{S,\min},\data)=1-F_{1,m}(\hatt\kappa_S^2)$,
is already above 0.50, then the median confidence estimate
is equal to this minimum start value $\rr_{S,\min}$. 
Figure \ref{figure:figure15} shows the 31 confidence curves
of (\ref{eq:cdforrr}), indicating also the 
median confidence estimates. 
We may also read off confidence intervals, 
e.g.~$\{\rr_S\colon C_S(\rr_s,\data)\le0.90\}$,
to check which models are associated with 
high confidence for doing a better job than the wide model. 

Opening the FRIC box in this case, checking which 
submodels do particularly well for the case of Mrs.~Jones,
with the confidence distributions and median confidence
estimates, one learns that there is high and essential
agreement with the list of good models given 
in Table \ref{table:table11}. Also, the correlations 
between the three relative risks estimates 
proposed in this and the preceding section, 
i.e.~$\fric_S^u$, $\fric_S^t$, and now $\fric_S^{0.50}$, 
are all above 0.977. 

\section{AFRIC scores for variable selection} 
\label{section:afric}

\def\calU{{\cal U}}

Above we have developed versions of FRIC, with 
one particular focus parameter 
$\E\,(y\midd x_0)=x_0^\tr\beta$ at a time. Suppose
now that we take an interest in many such at the same
time, perhaps a stratum in the space of covariates,
like all white smoking mothers-to-be in the context
of the Mrs.~Jones example of our introduction. 
Consider covariate vectors $x(u)$ of interest, 
for an index set $u\in\calU$, along with a measure
of relative importance, say $v(u)$ for 
$\E\,(Y\midd x(u))=x(u)^\tr\beta$. 
These $x(u)$, and the importance weight $v(u)$,
are selected by the statistician, inside the 
given context, and they may or may not form 
a subset of the covariate vectors 
$\{x_1,\ldots,x_n\}$ associated with the observed data.
A natural combined measure of loss associated
with selecting submodel $S$ for the purpose of 
estimating all these $x(u)^\tr\beta$ parameters is 
\beqn
L_S=\sum_{u\in\calU} 
   v(u) \{x_S(u)^\tr\hatt\beta_S-x(u)^\tr\beta\}^2. 
\eeqn
By the efforts of Section \ref{section:basics} 
the associated risk is 
\beqn
\risk_S=\E\,L_S
   =(\sigma^2/n)\sum_{u\in\calU} v(u) 
   [x_S(u)^\tr\Sigma_{n,S}^{-1}x_S(u)+n\{\omega_S(u)^\tr\beta_{S^c}\}^2/\sigma^2],  
\eeqn
with 
$\omega_S(u)=\Sigma_{10,S}\Sigma_{00,S}^{-1}x_S(u)-x_{S^c}(u)$
for $u\in\calU$. In particular, for the wide model the risk is 
\beqn
\risk_\wide
   =(\sigma^2/n)\sum_{u\in\calU} v(u) 
   x(u)^\tr\Sigma_n^{-1}x(u), 
\eeqn 
leading to a clear relative risk ratio of total risks, 
\beq
\label{eq:rrforafric}
\rr_S={\risk_S\over \risk_\wide}
={\sum_{u\in\calU} v(u)[x_S(u)^\tr\Sigma_{n,S}^{-1}x_S(u)
   +n\{\omega_S(u)^\tr\beta_{S^c}\}^2/\sigma^2]
   \over \sum_{u\in\calU} v(u) x(u)^\tr\Sigma_n^{-1}x(u)}. 
\eeq 

It is useful to work out a clearer expression for the second
term in the numerator, also when it comes to the task
of estimating the full $\rr_S$ from data. We may write 
\beqn
\gamma_S
=\sum_{u\in\calU} v(u){n\over \sigma^2} \beta_{S^c}^\tr\omega_S(u)\omega_S(u)^\tr\beta_{S^c}
={n\over \sigma^2}\beta_{S^c}^\tr A_S\beta_{S^c}, 
\eeqn 
with $A_S=\sum_{u\in\calU} v(u)\omega_S(u)\omega_S(u)^\tr$. 
Since $\rootn\hatt\beta_{S^c}/\sigma \sim \N_{p-|S|}(\rootn\beta_{S^c}/\sigma,Q_S)$, 
with $Q_S$ as in (\ref{eq:hereisQS}), we find that 
the natural start estimator 
$\tilda\gamma_S=n\hatt\beta_{\wide,S^c}^\tr A_S\hatt\beta_{\wide,S^c}/\hatt\sigma^2$
has mean 
\beqn
\E\,\tilda\gamma_S
={m\over m-2}
   \Tr\{ A_S(n\beta_{S^c}\beta_{S^c}^\tr/\sigma^2+Q_S) \}
={m\over m-2}\{\gamma_S+\Tr(A_SQ_S)\},
\eeqn 
Two natural estimators of the $\gamma_S$ quantity are therefore 
\beqn
\hatt\gamma_S^u={m-2\over m}\{\tilda\gamma_S-\Tr(A_SQ_S)\}
\quadandquad 
\hatt\gamma_S^t={m-2\over m}\max\{0,\tilda\gamma_S-\Tr(A_SQ_S)\}, 
\eeqn 
the unbiased estimator and its truncated-to-zero version.
This leads to averaged FRIC scores, say 
\beq
\label{eq:africu}
\afric_S^u
=\hatt\rr_S^u
={\sum_{u\in\calU} v(u)x_S(u)^\tr\Sigma_{n,S}^{-1}x_S(u)
   +\hatt\gamma_S^u 
   \over \sum_{u\in\calU} v(u) x(u)^\tr\Sigma_n^{-1}x(u)}, 
\eeq
and an accompanying $\afric_S^t$ with the truncated version 
$\hatt\gamma_S^t$ instead of the unbiased one. 
For the given collection of $x(u)$ covariate vectors,
along with importance numbers $v(u)$, one may now 
compute the $\afric_S$ scores for all candidate models,
with the best models those with smallest such scores. 
 
\begin{figure}[h]
\centering
\includegraphics[scale=0.55]{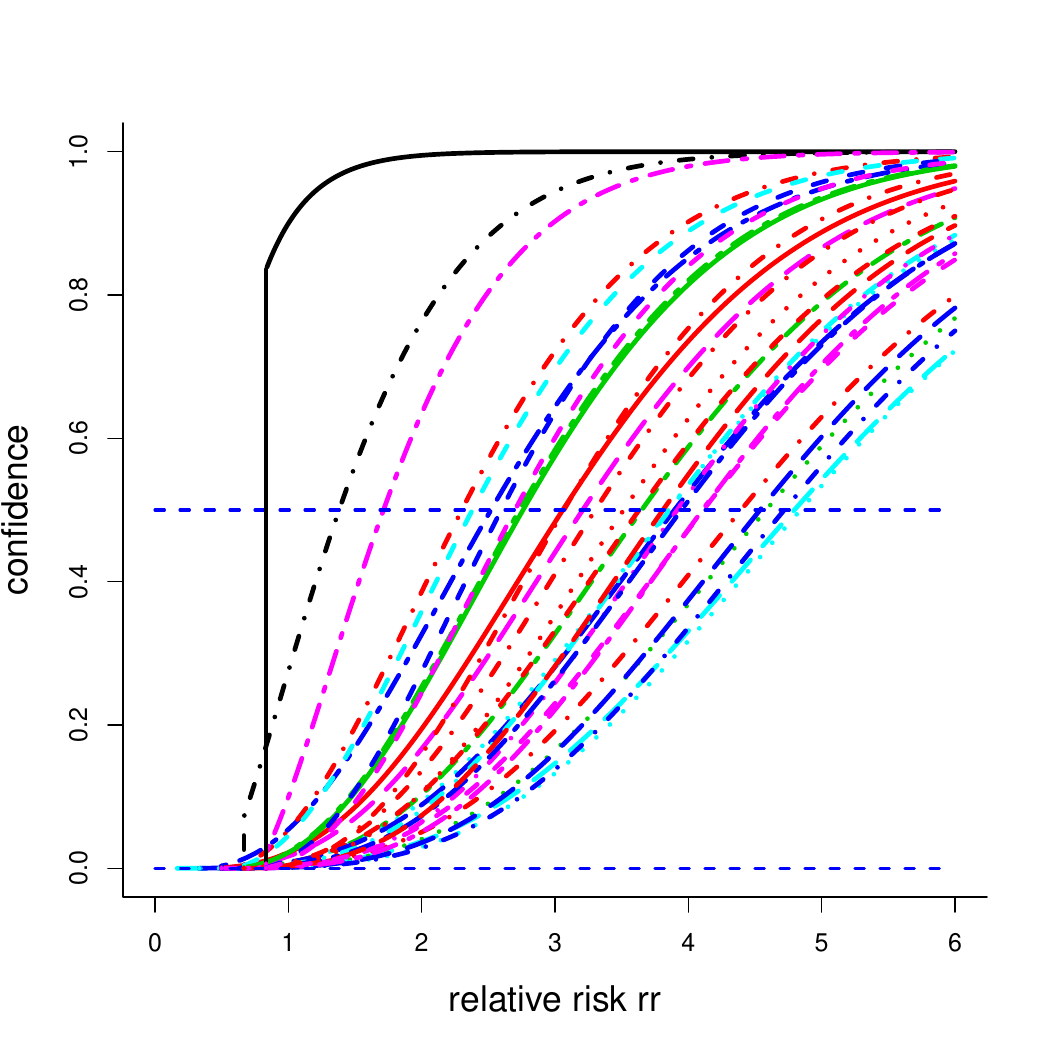}
\caption{Confidence cumulative distribution functions
$C^*_S(\rr_S)$ for the relative risks $\risk_S/\risk_\wide$,
for the 31 submodels for the mothers-and-babies data,
for the case of putting equal importance to all available
covariate vectors, i.e.~using $v(u)=1/n$ for the $n$ available
vectors.  
Only the submodel corresponding to keeping $x_2,x_3,x_4,x_5$,
but excluding $x_1$, with the black full curve, 
exhibits a clear confidence for working better than 
the wide model. The $\afric^{0.50}_S$ scores of (\ref{eq:africm}) 
are read off where the confidence curves cross $0.50$. }
\label{figure:figure14}
\end{figure}

An important special case of the AFRIC strategy is to 
declare all available covariate vectors $x_i$ 
equally important, with importance weight $1/n$. 
For this case we have 
$\sumin n^{-1} x_i^\tr\Sigma_n^{-1}x_i
   =\Tr(\Sigma_n^{-1}\Sigma_n)=p$, and similarly 
\beqn
\sumin n^{-1}x_{i,S}^\tr\Sigma_{n,S}x_{i,S}
   =\Tr(\Sigma_{n,S}^{-1}\Sigma_{n,S})=|S|, 
\eeqn 
the cardinality of $S$. 
Also, the $\omega_S$ vector for $x_i$ is  
$\omega_{i,S}=\Sigma_{10,S}\Sigma_{00,S}^{-1}x_{i,S}-x_{i,S^c}$. 
The second term of the numerator for 
the relative risk (\ref{eq:rrforafric}) may hence be written 
$\gamma_S=n\beta_{S^c}^\tr A_S\beta_{S^c}/\sigma^2$, with 
\beqn
A_S&=&\sumin n^{-1} (\Sigma_{10,S}\Sigma_{00,S}^{-1}x_{i,S}-x_{i,S^c})
                 (\Sigma_{10,S}\Sigma_{00,S}^{-1}x_{i,S}-x_{i,S^c})^\tr \\ 
&=&\Sigma_{10,S}\Sigma_{00,S}^{-1}\Sigma_{00,S}\Sigma_{00,S}^{-1}\Sigma_{01,S}
   -2\,\Sigma_{10,S}\Sigma_{00,S}^{-1}\Sigma_{01,S}
   +\Sigma_{11,S} \\
&=& \Sigma_{11,S}-\Sigma_{10,S}\Sigma_{00,S}^{-1}\Sigma_{01,S} 
 =Q_S^{-1}, 
\eeqn
see (\ref{eq:hereisQS}). We learn that with equal importance for 
all covariate vectors, the relative risk illuminatingly 
can be written 
\beq
\label{eq:rrSandgammaS}
\rr_S
={\risk_S\over \risk_\wide} 
={|S|+\gamma_S \over p} 
={|S|+n\beta_{S^c}^\tr Q_S^{-1}\beta_{S^c}/\sigma^2 \over p}. 
\eeq 
With $\tilda\gamma_S=n\hatt\beta_{\wide,S^c}^\tr Q_S^{-1}
   \hatt\beta_{\wide,S^c}/\hatt\sigma^2$
the AFRIC scores above simplify to 
\beq
\label{eq:africspecial}
\afric_S^u
=\hatt\rr_S^u
= \Bigl[|S|+{m-2\over m}\{\tilda\gamma_S-(p-|S|)\} \Bigr] / p, 
\eeq 
with a sister version using truncation to zero 
for the second term in the case of $\tilda\gamma_S<p-|S|$. 
These scores are easily computed, yielding a complete ranking
of all candidate models, for those with smallest 
estimated relative risks to the highest. 

For this special case, when all available covariates 
are considered equally important, we also learn that 
a submodel $S$ is better than the wide one provided 
\beqn
\gamma_S=n\beta_{S^c}^\tr Q_S^{-1}\beta_{S^c}/\sigma^2\le p-|S|. 
\eeqn 
This can be assessed accurately via 
\beqn
\tilda\gamma_S
={n\hatt\beta_{\wide,S^c}^\tr Q_S^{-1}\hatt\beta_{\wide,S^c}\over \hatt\sigma^2}
\sim {\chi^2_{p-|S|}(\gamma_S)\over \chi^2_m/m}
\sim(p-|S|)F_{p-|S|,m}(\gamma_S), 
\eeqn
featuring the noncentral $F$ with degrees of freedom
$(p-|S|,m)$ and excentre parameter $\gamma_S$. 
This leads to a full and exact confidence distribution
for the $\gamma_S$ parameter, and via the representation 
(\ref{eq:rrSandgammaS}) above also for the relative risk 
$\rr_S$ itself: 
\beq
\label{eq:cdforafric}
\begin{array}{rcl}
C^*_S(\rr_S)
&=&\displaystyle 
\Pr_{\gamma_S}\{ (|S|+\tilda\gamma_S)/p \ge (|S|+\tilda\gamma_{S,\obs})/p\} \\
&=&\displaystyle 
1-F_{p-|S|,m}\Bigl({\tilda\gamma_{S,\obs}\over p-|S|},
   p\,\rr_S-|S|\Bigr)
\quad {\rm for\ }\rr_S\ge |S|/p. 
\end{array}
\eeq 
Note that these start at minimal $\rr_S$ values, namely 
$|S|/p$, with start confidence pointmasses equal to 
$1-F_{p-|S|,m}(\tilda\gamma_{S,\obs}/(p-|S|))$. 

There are several uses of these confidence distributions
for the relative risks. First, we may use the 
median confidence estimator to create a different 
AFRIC score:  
\beq
\label{eq:africm}
\afric_S^{0.50}=(C^*_S)^{-1}(0.50)
   =\min\{\rr_S\ge|S|/p\colon C^*_S(\rr_S)\ge0.50\}. 
\eeq 
This version bypasses the difficulties related to 
truncation-to-zero or not. Second, we may read off
the post-data confidence probabilities that 
given submodels work better than the wide model, via 
\beqn
C^*_S(1)
=\Pr^*\{\rr_S\le 1\}
=1-F_{p-|S|,m}\Bigl({\tilda\gamma_{S,\obs}\over p-|S|},p-|S|\Bigr). 
\eeqn 

Figure \ref{figure:figure14} displays the confidence 
distribution functions $C^*_S(\rr_S)$ for all the 
relative risks $\risk_S/\risk_\wide$ for the dataset described
in the introduction, with five covariates $x_1,\ldots,x_5$
recorded to assess their potential influence on 
birthweight $y$. Using focused FRIC methods developed in
Sections \ref{section:basics}--\ref{section:rrwithCD} 
we have seen that for various given purposes, 
like predicting the outcome for a given mother-to-be, 
there might be several submodels doing far better 
than the wide model; see 
Figures \ref{figure:figure11}--\ref{figure:figure13}.
Using the AFRIC apparatus, for the case of all available
covariates deemed to have equal importance, 
Figure \ref{figure:figure14} indicates that most submodels
can be expected to do worse than the full wide model,
however. Only one model, the one keeping $x_2,x_3,x_4,x_5$
but excluding $x_1$ (age of mother), is then scoring 
significantly better than the wide model. This is also
shown by inspecting the AFRIC scores, where 
those of (\ref{eq:africspecial}) and (\ref{eq:africm})
turn out to be very highly correlated. 

\section{AFRIC, the AIC, and the Mallows criterion}
\label{section:mallows}

The FIC and FRIC methods are indeed focused, set up
to work well with a given focus parameter, and 
are also for that reason different in spirit 
from overall variable and model selection methods
like the AIC, the BIC, the Mallows statistic, etc. 
The particular AFRIC method of (\ref{eq:africspecial})
is however an intended de-focused overall selection
criterion, where all covariate vectors in the data 
collection of $(x_i,y_i)$ are considered equally important, 
and below we make some comparisons
with other methods. For general material on these
selection criteria, see \citet{ClaeskensHjort08}. 

\subsection{The AIC}

The log-likelihood for the $S$ subset model takes the form 
\beqn
\ell_{n,S}(\beta_S,\sigma_S)
   =-n\log\sigma_S-\half(1/\sigma_S^2)\sumin (y_i-x_{i,S}^\tr\beta_S)^2
   -\half n\log(2\pi), 
\eeqn 
with maximisers $\hatt\beta_S$ of (\ref{eq:betahatS}) and 
\beq
\label{eq:rss}
\hatt\sigma_S^2={\rss(S)\over n},
\quad {\rm with} \quad 
\rss(S)=\sumin (y_i-x_{i,S}^\tr\hatt\beta_S)^2. 
\eeq 
Often the unbiased version $\rss(S)/(n-|S|)$ is used
instead, but $\hatt\sigma_S^2$ with denominator $n$ is 
the maximum likelihood estimator. The log-likelihood 
maximum is hence 
$\ell_{n,\max}=-n\log\hatt\sigma_S-\half n-\half n\log(2\pi)$,
so that the AIC score becomes 
\beqn
\AIC_S
=2\,\ell_{n,S,\max}-2\,(|S|+1)
=-2n\log\hatt\sigma_S-2|S|+c_n, 
\eeqn 
with $c_n$ an immaterial constant not depending on 
the data or the candidate model. In this formulation
models with higher AIC values are preferred, 
so the AIC hence selects the model with the lowest score 
$n\log\hatt\sigma_S+|S|$. 

There are several variations on the AIC scheme, 
including finite-sample corrections and other
approximations, see \citet[Ch.~2]{ClaeskensHjort08}. 
One version is related to a simple Taylor expansion,
valid when the $\hatt\sigma_S/\hatt\sigma_\wide$ ratios
are not far from 1. Starting with 
\beqn
\AIC_S
=-n\log\Bigl(\hatt\sigma_\wide^2{\hatt\sigma_S^2\over \hatt\sigma_\wide^2}\Bigr)
   -2|S|+c_n 
=-n\log\Bigl(1+{\hatt\sigma_S^2\over \hatt\sigma_\wide^2}-1\Bigr)
    -2|S|+c_n'
\eeqn 
we reach the approximation $\AIC_S\doteq\AIC_S^*$, where 
\beq
\label{eq:aicstar} 
\AIC_S^*=-n{\hatt\sigma_S^2\over \hatt\sigma_\wide^2}-2|S|+c_n''
=-{\rss(S)\over \hatt\sigma_\wide^2}-2|S|+c_n'', 
\eeq 
with $c_n',c_n''$ further immaterial constants. 
Thus this approximate AIC selects the set $S$ with smallest 
$n\hatt\sigma_S^2/\hatt\sigma_\wide^2+2|S|
=n\,\rss(S)/\rss(\wide)+2|S|$. 
This version of AIC is also arrived at if one starts
with the modelling setup that for submodel $S$,
one has $y_i\sim\N(x_{i,S}^\tr\beta_S,\sigma^2)$, 
with the same $\sigma$ across candidate models. 
This is arguably conceptually wrong, as models 
with more covariates ought to have smaller and
not identical residual variances, but as the Taylor argument 
shows the error is slight if the 
$\hatt\sigma_S/\hatt\sigma_\wide$ ratios are not too big. 
There are connections to the AFRIC 
of (\ref{eq:africspecial}), which we comment on below,
after noting that $\AIC^*$ is also equivalent 
to the Mallows Cp method. 

\subsection{The Mallows statistic}

There are several equivalent or close-to-equivalent 
definitions of the Mallows Cp statistic, 
from \citet{Mallows73, Mallows95}; 
see e.g.~\citet{Hansen07} and \citet[Ch.~12]{EfronHastie16}
for somewhat different perspectives and types of uses. 
For the present purposes, for each submodel $S$, let 
\beq
\label{eq:mallows}
M_S
=\rss(S)/\hatt\sigma^2-n+2|S|, 
\eeq
where the tradition is to use the unbiased 
$\hatt\sigma^2=\rss(\wide)/(n-p)$ in the denominator,
as opposed to the maximum likelihood version 
$\hatt\sigma_\wide^2=\rss(\wide)/n$ encountered 
for the AIC above. With this definition, we have 
$M_\wide=n-p-n+2p=p$ for the wide model. 
The difference between the two variance estimators 
is small in the traditional setup with $p$ small or moderate
and $n$ moderate or large. So apart from this often
minor detail, the Mallows score (\ref{eq:mallows}) 
is just a minus sign and a constant away from 
the AIC approximation (\ref{eq:aicstar}). 
So high scores $\AIC^*_S$ are equivalent to low scores $M_S$. 

There is a certain connection from the Mallows score
and hence also the AIC approximation  
to the neutral-version AFRIC score, 
which we now look into. It will become apparent that 
the AFRIC score (\ref{eq:africspecial}) can be seen 
as a more sophisticated cousin of Mallows. Note that 
with $1-2/m$ approximated as 1 we have  
$\afric_S$ equal to 
\beq
\label{eq:africapprox} 
\tilda\gamma_S+2|S|
=n\hatt\beta_{\wide,S^c}^\tr Q_S^{-1} \hatt\beta_{\wide,S^c}/\hatt\sigma^2
   +2|S|, 
\eeq 
up to constants, 
which has a clear structural similarity to the
$\rss(S)/\hatt\sigma^2+2|S|$ with Mallows.  
Also, both $\tilda\gamma_S$ and $\rss(S)/\hatt\sigma^2$ 
relate to the size of the mean parameter modelling bias,
the $\beta_{S^c}$, but this modelling bias is being
assessed differently by the two methods, and 
actually more precisely via the AFRIC. 

We now examine the required details, to learn the extent
to which low scores of $\tilda\gamma_S+2|S|$ 
is related to or sometimes equivalent to low scores 
for $\rss(S)/\hatt\sigma^2+2|S|$. 
The nutshell story for the neutral AFRIC, as developed 
in Section \ref{section:afric}, is 
(i) that the crucial quantity to estimate is $|S|+\gamma_S$, 
with $\gamma_S=n\beta_{S^c}^\tr Q_S^{-1}\beta_{S^c}/\sigma^2$, and 
(ii) that when using the estimator 
$\tilda\gamma_S=n\hatt\beta_{\wide,S^c}^\tr Q_S^{-1}
   \hatt\beta_{\wide,S^c}/\hatt\sigma^2$, and 
correcting for bias, the result is 
$|S|+\tilda\gamma_S-(p-|S|)$, modulo approximating 
$1-2/m$ by 1. The parallel nutshell story for the 
Mallows criterion must be presented differently,
in that step (ii), the statistic itself, 
comes before the insight of type (i), 
uncovering the underlying crucial population parameter
being estimated. It turns out to be the same 
as for the neutral AFRIC. 

\begin{lemma}{{\rm 
Consider the Mallows variable $M_S^0=\rss(S)/\sigma^2-n+2|S|$,
here  using the real $\sigma$ rather than the estimated one,
as with $M_S$ of (\ref{eq:mallows}). Its mean is $|S|+\gamma_S$.}}
\label{lemma:lemma12}
\end{lemma} 


\begin{proof}
Note first that the textbook formula $\E\,\rss(S)=(n-|S|)\sigma^2$
does not apply here, since submodel $S$ does not hold,
but has a mean modelling bias. 
There are several formulae within reasonable reach 
for the mean of $\rss(S)$, under the conditions of 
the wide model (\ref{eq:widemodel}), 
but certain linear algebra efforts are required 
to prove that is can be expressed as 
$|S|+\gamma_S$, with the same $\gamma_S$ as 
for the AFRIC. 

In Section \ref{section:basics} we wrote the basic
model in terms of the individual $y_i=x_i^\tr\beta+\eps_i$, 
and now it is practical to also work with the 
linear algebra version, writing 
$y=X\beta+\eps\sim\N_n(X\beta,\sigma^2 I)$,
with $X$ the $n\times p$ matrix of the covariate vectors.  
We start with 
\beqn
y-X_S\hatt\beta_S=\{ I-X_S (X_S^\tr X_S)^{-1}X_S^\tr \}y=(I-H_S)y, 
\eeqn  
with $H_S=X_S(X_S^\tr X_S)^{-1}X_S^\tr$, 
the idempotent hat matrix with $H_S^2=H_S$ and $\Tr(H_S)=|S|$. 
For the mean vector of the data vector $y$,
write $\xi=X\beta=X_S\beta_S+X_{S^c}\beta_{S^c}$, so that 
$b_S=\E\,y-X_S\beta_S=X_{S^c}\beta_{S^c}$ 
is the mean modelling error when employing submodel $S$. 
We now have 
\beqn
\xi^\tr(I-H_S)\xi
&=&(X_S\beta_S+X_{S^c}\beta_{S^c})^\tr (I-H_S) 
 (X_S\beta_S+X_{S^c}\beta_{S^c}) \\
&=&(X_S\beta_S+X_{S^c}\beta_{S^c})^\tr (I-H_S)X_{S^c}\beta_{S^c} \\
&=&b_S^\tr(I-H_S)b_S, 
\eeqn 
which leads to 
\beqn
\E\,\rss(S)
=\E\,y^\tr (I-H_S)y 
&=&\xi^\tr (I-H_S)\xi+\Tr \{(I-H_S)\,\Var\,y\} \\
&=&b_S^\tr(I-H_S)b_S+\sigma^2(n-|S|).
\eeqn
We have reached 
$\E\,M_S^0=n-|S|+\phi_S-n+2|S|=|S|+\phi_S$, with 
\beqn
\phi_S
=b_S^\tr(I-H_S)b_S/\sigma^2
=\beta_{S^c}X_{S^c}^\tr(I-H_S)X_{S^c}\beta_{S^c}/\sigma^2. 
\eeqn 
But this is verifiably the same as $\gamma_S$, 
in that the matrix identity 
\beq
\label{eq:matrixidentity}
n^{-1}X_{S^c}^\tr(I-H_S)X_{S^c}
=Q_S^{-1}
=\Sigma_{11,S}-\Sigma_{10,S}\Sigma_{00,S}^{-1}\Sigma_{01,S}
\eeq 
can be proved separately, 
with $Q_S$ of (\ref{eq:hereisQS}). 
\end{proof}

The exact mean of the Mallows statistic $M_S$ 
of (\ref{eq:mallows}) is a more complicated expression
than for the $M_S^0$ of the lemma, but 
since $\hatt\sigma/\sigma$ is close to 1 with 
high probability, for $m=n-p$ moderate or large, 
we would still have $|S|+\gamma_S$ as the 
approximate mean of $M_S$. 

The neutral AFRIC of (\ref{eq:africspecial}) 
has a couple of extra sales points when compared to 
the partly similar Mallows criterion. It is exact,
and properly fine-tuned for finite samples, with no
approximations involved; also, it has a clear statistical
interpretation, as an exactly unbiased estimator 
of overall relative risk, the $\rr_S=\risk_S/\risk_\wide$
of (\ref{eq:rrSandgammaS}). With the AFRIC formulation, 
part of the usefulness of the strategy is to 
only care about submodels with AFRIC scores less 
than 1. Finally there is a clear, 
exact, and optimal associated confidence distribution, 
the $C_S^*(\rr_S)$ of (\ref{eq:cdforafric}). 


\section{Concluding remarks} 
\label{section:concluding}

We round off our paper by offering a list of 
concluding remarks, pointing both to related themes, 
further applications of the developed methods, 
and to a few issues for further studies.  

\medskip
{\it A. The birthweight dataset.} 
The dataset on birthweights for 189 babies, 
with covariate information for their mothers, 
stems from \citet{HosmerLemeshow89}, and has 
later been used by several authors for reanalyses
and illustration of new methodology; 
see e.g.~\citet{ClaeskensHjort08}. In nearly all of 
these publication the emphasis has been on 
the event `lower than 2.50 kg or not', however, 
i.e.~with logistic regressions with variations, 
whereas we here actually use the actual birthweights 
on their continuous scale.   

Our statistical narrative has focused Mrs.~Jones
(age 40, weight 60, white, smoker),  
to convey that the methodology is intended to work
for one person or object at the time; were she 
not a smoker we could easily re-run all programmes 
to produce FRIC plots and FRIC tables for sorting
and ranking all candidate models again, indeed
finding a different `best model' for that purpose. 
 
\medskip
{\it B. Candidate models.}
Our variable selection machinery is able to handle 
comparisons between all subset driven models, 
as seen in the introduction example with Mrs.~Jones
and the ensuing $2^5=32$ candidate models for her 
five covariates. It is useful if the statistician
can carry out an initial screening, however, 
omitting implausible models; one might e.g.~allow 
interaction terms $x_jx_k$ to enter a candidate model 
only when both $x_j$ and $x_k$ are on board. Sometimes 
there is also a natural ordering of models, 
as in nested situations. Screening away implausible
models makes the rest of the FRIC and AFRIC 
analyses simpler and sharper. 

A pertinent reminder is that results developed 
and used in our paper do rely on the starting assumption 
that the wide model (\ref{eq:widemodel}) holds. 
This might be checked separately, by any of battery 
of goodness-of-fit checks for the linear regression model.
The more critical of these underlying assumptions
are independence and constancy of error variance;
our methods will continue to work well even
if residuals are not exactly normally distributed. 

\medskip
{\it C. A single submodel versus the wide.}
A special case of the FRIC and AFRIC setups is when
there is a single submodel $S$ to check against 
the wide model. Using Lemma \ref{lemma:lemma11} 
we reached the characterisation (\ref{eq:betterornot}),
that $S$ does better than the wide if and only if 
$\rootn|\omega_S^\tr\beta_{S^c}|<(\omega_S^\tr Q_S\omega_S)^{1/2}$.
This is an infinite strip in the space of $\beta_{S^c}$, 
of those not far from perpendicularity to $\omega_S$,
and this depends on the $x_0$ under focus. 
The FRIC methods of Sections \ref{section:basics}--\ref{section:rrwithCD} 
are built to take this possibility into account, 
that model $S$ can perform well even if it is 
far from correct, and, specifically, even if it 
might be doing a bad job for another $\E\,(Y\midd x_0^*)$. 
This makes the FRIC methods different from those which 
directly or indirectly involve testing whether 
$\beta_{S^c}$ is zero or close to zero. The best focused 
tool for comparing $S$ with the wide is the 
full confidence curve $C_S(\rr_S,\data)$ of (\ref{eq:cdforrr}),
see Figure \ref{figure:figure15}. 

The neutral or de-focused AFRIC scheme developed 
in Section \ref{section:afric} is different, however,
and when all covariate vectors are seen as equally
important matters are seen to hinge on the `global' 
statistic $\tilda\gamma_S=n\hatt\beta_{S^c}^\tr Q_S^{-1}\hatt\beta_{\wide,S^c}$.
The best tool for comparing $S$ with the wide 
is now the confidence distribution $C_S^*(\rr_S)$ 
of (\ref{eq:cdforafric}) for the overall relative risk 
$\rr_S=\risk_S/\risk_\wide$. 

\begin{figure}[h]
\centering
\includegraphics[scale=0.55]{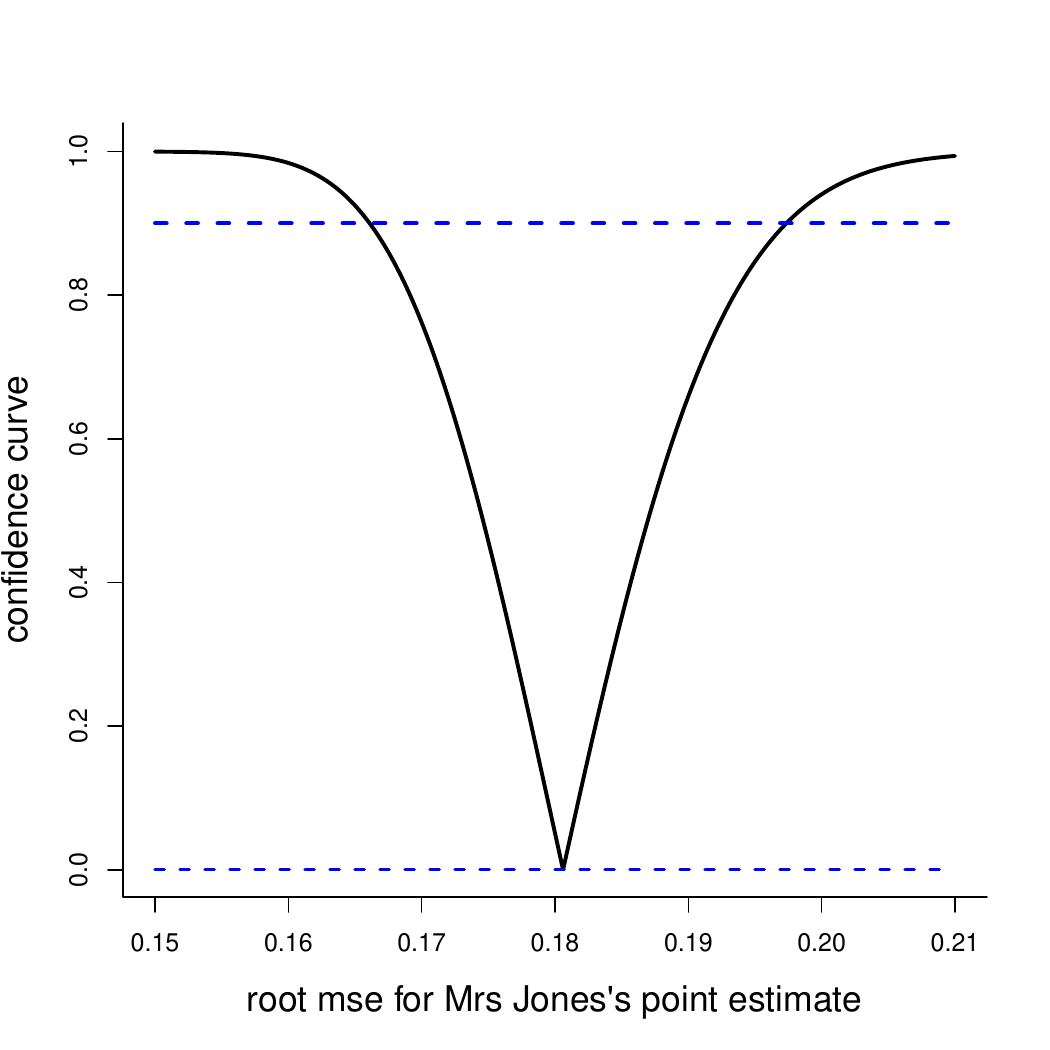}
\caption{Confidence curve for the root mean squared error
of the estimator $\hatt\mu_\wide=x_0^\tr\hatt\beta_\wide$
for the case of Mrs.~Jones. The estimate itself 
is 2.889 kg, and the estimated root-mse of that estimate
is 0.180, with the figure displaying the full $\cc(\rmse_\wide)$.}
\label{figure:figure16}
\end{figure}

\medskip
{\it D. Model averaging.}
A natural follow-up to the present study is that 
of constructing good model averaging procedures, 
of the type $\hatt\mu^*=\sum_S w(S)\hatt\mu_S$, 
perhaps with data-driven weights summing to one
over the different submodel based choices $\hatt\mu_S$. 
These weights could reflect how well the different 
candidate models do according to the $\fric(S)$
or $\conf_S$ criteria. This would relate to similar
strategies studied in \citet{Hansen07}, 
\citet{Claeskens16}, and \citet{CunenHjort20a}. 

\medskip
{\it E. rmse and mse.}
We have given our estimators and confidence curves for 
the scale of relative risks, $\rr_S=\mse_S/\mse_\wide$; 
also, the $\afric_S=\hatt\rr_S$ works by intention on
that same scale. It is sometimes more meaningful to work 
in the root-relative-risk scale of $\rrr_S=\rmse_S/\rmse_\wide$,
however, since the root-mse are on the scale of the 
focus parameter. This is a matter of choice and convenience; 
all relevant numbers, and hence also the FRIC and Confidence plots, 
are easily converted to the $\rrr_S$ scale if wished for. 
For the $2^5=32$ submodels encountered for Mrs.~Jones
we chose $\rr_S$ for visual clarity with the FRIC plot; 
both the best and the worst candidate models are 
easier to spot with Figure \ref{figure:figure11} 
than with a corresponding $\fric^{1/2}$ plot. 
Incidentally, the Confidence plot of Figure \ref{figure:figure13}
is not affected by the scale used for relative risks. 
 
In addition to letting FRIC do accurate assessments
of all relative risks $\rmse_S/\rmse_\wide$, 
an analysis might not be quite complete without also assessing 
the $\rmse_\wide$ itself, the root-mse of the estimator 
$\hatt\mu_\wide=x_0^\tr\hatt\beta_\wide$. A formula for 
this quantity is $\rmse_\wide=(\sigma/\rootn)(x_0^\tr\Sigma_n^{-1}x_0)^{1/2}$,
and both estimation and accurate assessment are now 
easily available through $\hatt\sigma\sim\sigma(\chi^2_m/m)^{1/2}$. 
For the case of Mrs.~Jones, estimating $\mu=\E\,(Y\midd x_{\rm jones})$
via the wide model gives $\hatt\mu_\wide=2.889$, 
its estimated root-mse is 0.180, and a full confidence curve is 
\beqn
\cc(\rmse_\wide)=\Big| 1
   -2\,\Gamma_m\Bigl(m{\hatt\mse_\wide\over \mse_\wide}\Bigr) \Big|,
\eeqn  
here with $\hatt\mse_\wide$ observed at $0.180^2$, 
and with $\Gamma_m(\cdot)$ the distribution function for the $\chi^2_m$. 
This is displayed in Figure \ref{figure:figure16}, 
pointing to the median confidence estimate, 
and with confidence intervals easily read off for 
each wished-for confidence level. This confidence curve
is the optimal one, see \citet[Ch.~6]{SchwederHjort16}.
The FRIC analysis of Section \ref{section:intro} 
indicates that there are several submodel based 
estimates of $\E\,(Y\midd x_{\rm jones})$ with smaller
root-mse than 0.180. 

It might of course also be of interest to assess 
the $\rmse_S$ quantities directly. This is 
indeed approachable, via 
\beqn
\mse_S
=\mse_\wide\,\rr_S
=(\sigma^2/n)
   (x_{0,S}^\tr\Sigma_{n,S}^{-1}x_{0,S} + \omega_S^\tr Q_S\omega_S\,\kappa_S^2), 
\eeqn 
using Lemma \ref{lemma:lemma11} again. 
The technical obstacle is that estimators of $\sigma$
and $\kappa_S$ are dependent, see (\ref{eq:kappakappa}), 
and there is no exact confidence curve for the product 
expression of $\mse_S$. Good approximations may be worked
out, via techniques of \citet[Ch.~3-4]{SchwederHjort16}, 
but our focus on the $\rr_S$ has bypassed the need for 
such approximations.  

\medskip
{\it F. When the $\Sigma_n$ matrix is diagonal.} 
Suppose the variance matrix $\Sigma_n=n^{-1}\sumin x_ix_i^\tr$
is equal to the identity matrix. 
Then the components of $\hatt\beta_\wide=n^{-1}\sumin x_iy_i$ 
are independent and $\hatt\beta_{\wide,j}\sim\N(\beta_j,\sigma^2/n)$. 
The structure of FRIC and AFRIC formulae simplify and it also 
becomes easier to examine aspects of performance. 
Considering first the AFRIC of (\ref{eq:africspecial}), we find 
\beqn
\afric_S^u
=\Bigl[|S|+(1-2/m)\Bigl\{ \sum_{j\notin S}
   n\hatt\beta_{\wide,j}^2/\hatt\sigma^2-(p-|S|)\Bigr\} \Bigr]/p,  
\eeqn 
where the numerator may be written 
\beqn
\sum_{j\in S} 1 + \sum_{j\notin S}\hatt\phi_j
=\sum_{j=1}^p [I\{j\in S\}+\hatt\phi_j\,I\{j\notin S\}],
\quad {\rm where} \quad 
\hatt\phi_j={m-2\over m}
   \Bigl({n\hatt\beta_{\wide,j}^2\over\hatt\sigma^2}-1\Bigr). 
\eeqn 
The winning subset $S^*$, where this becomes smallest,
is where $j$ is selected if and only if $\hatt\phi_j>1$,
which means 
$n\hatt\beta_{\wide,j}^2/\hatt\sigma^2>2+2/(m-2)$. 
This corresponds to coordinate-wise inspection for 
non-zero-ness of the $\beta_j$, with a significance level 
of about $\Pr\{\chi^2_1>2\}=0.157$, for large $n$.  
For the FRIC of (\ref{eq:fricu}) we similarly find 
\beqn
\fric_S^u={\sum_{j\in S}x_{0,j}^2
   +(1-2/m)\sum_{j\notin S}x_{0,j}^2 (n\hatt\beta_{\wide,j}^2/\hatt\sigma^2-1)
   \over \sum_{j=1}^p x_{0,j}^2}. 
\eeqn 
The numerator can be written 
$\sum_{j=1}^p x_{0,j}^2[I\{j\in S\}+\hatt\phi_j\,I\{j\notin S\}]$,
leading to precisely the same winner $S^*$ as just found
for the AFRIC case, i.e.~independent of the $x_0$ in question.
This is a consequence of the diagonal assumption for $\Sigma_n$,
however, and in general the best FRIC model very much
depends on the $x_0$ under scrutiny. 

\medskip
{\it G. More general regression models.}
The present paper has worked inside the wide linear 
regression model (\ref{eq:widemodel}), leading to 
precise finite-sample confidence distributions for 
all relative risks, for any focused mean parameter 
$\E\,(Y\midd x_0)$. For more general regression models, 
as with the generalised linear models variety, and 
for general focus parameters, like above-threshold 
probabilities $\Pr\{Y\ge y_0\midd x_0\}$, 
one must rely on approximations, via large-sample arguments 
or bootstrapping. Such are worked out in \citet{CunenHjort20a},
with applications to logistic and Poisson regression setups; 
it is in general not possible to reach exact finite-sample
coverage outside the linear regression model, however. 

\section*{Acknowledgments} 

This article has benefitted from fruitful conversations
with C\'eline Cunen and Emil Stoltenberg, 
and from discussions with other members of 
the Norwegian Research Council funded project 
FocuStat: Focus Driven Statistical Inference 
with Complex Data, at the Department of Mathematics,
University of Oslo.
 
{{\small
\bibliographystyle{biometrika}
\bibliography{fic_bibliography2020.bib}
}}

\end{document}